\documentclass[letterpaper, 10pt, conference]{ieeeconf}  
\IEEEoverridecommandlockouts 
\usepackage{graphicx, subcaption, xcolor, cite} 
\definecolor{NCSUred}{RGB}{153, 0, 0}
\definecolor{NCSUgreen}{RGB}{0, 132, 115}
\definecolor{NCSUblue}{RGB}{65, 86, 161}
\definecolor{NCSUorange}{RGB}{209, 73, 5}

\newcommand{\mc}[1]{\mathcal{#1}}

\newcommand{\mr}[1]{\mathrm{#1}}

\newcommand{\mbb}[1]{\mathbb{#1}}
\newcommand{\rg}[1]{\mathring{#1}}
\newcommand{\mR}{\mathbb{R}}

\newcommand{\mN}{\mathbb{N}}
\newcommand{\xD}[1]{\mr{d} #1}

\newcommand{\bra}[1]{\left( #1 \right)}
\newcommand{\Bra}[1]{\left[ #1 \right]}

\newcommand{\ip}[2]{\langle #1, \, #2 \rangle}
\newcommand{\vk}{\varkappa}

\newcommand{\vf}{\varphi}

\usepackage{amsmath, amssymb}
\usepackage{algorithm, algpseudocode}
\usepackage{enumerate}
\usepackage{hyperref}
\hypersetup{colorlinks=true, breaklinks=true, plainpages=false, citecolor=NCSUblue, linkcolor=NCSUblue, urlcolor=NCSUorange}
\newtheorem{theorem}{Theorem}
\newtheorem{definition}{Definition}
\newtheorem{remark}{Remark}

\newtheorem{assumption}{Assumption}
\newtheorem{lemma}{Lemma}
\newtheorem{corollary}{Corollary}

\title{\LARGE \bf Dissipativity Analysis of Nonlinear Systems: \\ A Linear--Radial Kernel-based Approach}
\author{Xiuzhen Ye, Wentao Tang
\thanks{This work was supported by NSF-CBET Award \#2414369 and Faculty Startup Fund from North Carolina State University.}
\thanks{The authors are with Department of Chemical and Biomolecular Engineering, North Carolina State University, U.S.A. Corresponding author: W. Tang ({\tt\small wtang23@ncsu.edu})}
}

\begin{document}
\maketitle\thispagestyle{empty}
\pagestyle{empty}
\begin{abstract}
Estimating the dissipativity of nonlinear systems from empirical data is useful for the analysis and control of nonlinear systems, especially when an accurate model is unavailable. 
Based on a Koopman operator model of the nonlinear system on a reproducing kernel Hilbert space (RKHS), the storage function and supply rate functions are expressed as \emph{kernel quadratic forms}, through which the dissipative inequality is expressed as a \emph{linear operator inequality}. 
The RKHS is specified by a \emph{linear--radial kernel}, which inherently encode the information of equilibrium point, thus ensuring that all functions in the RKHS are ``locally at least linear'' around the origin and that kernel quadratic forms are ``locally at least quadratic'', which expressively generalize conventional quadratic forms (including sum-of-squares polynomials). 
Based on the kernel matrices of the sampled data, the dissipativity estimation can be posed as a finite-dimensional convex optimization problem, and a statistical learning bound can be derived on the kernel quadratic form for the probabilistic approximate correctness of dissipativity estimation. 
\end{abstract}

%%%%%%%%%%%%%%%%%%%%%%%%%%%%%%%%%%%%%%%%%%%%%%%%%%%%%%%%%%%%%%%%%%%%%%%%%%%%%%%%
\section{INTRODUCTION}
Dissipativity \cite{lozano2013dissipative}, defined as the property that the change in a state-dependent storage function does not exceed a supply rate dependent on inputs and outputs, is an important aspect of nonlinear systems behaviors \cite{willems1997introduction}. 
Such a concept is naturally connected to stability concepts \cite{seiler2014stability, scherer2022dissipativity}, and gives a compositional input--output approach to analyze interconnected systems \cite{willems2002interconnections, arcak2016networks}. 
Naturally, since engineering systems are commonly governed by energy conservation and dissipation (as two fundamental laws of thermodynamics), the existence of dissipative properties can be found in robotic, electric, and chemical systems \cite{alonso1996process, bao2007process, van2014port}. 

\par The knowledge of dissipative property can be acquired via \emph{first-principles, model-based, or data-driven} approaches. 
The first-principles calculation of storage function and supply rate is useful when the nonlinear system is mechanistically interpretable, with analytical expressions for energy-like functions, or decomposable into such subsystems \cite[Ch. 6]{khalil2002nonlinear}. 
As for the model-based approach, while the dissipativity condition of nonlinear systems (generalized Kalman--Yakubovich--Popov Lemma) is provided as functional inequalities \cite{hill1976stability}, its applicability is usually restricted to linear cases, where the storage and supply rate are quadratic functions \cite{willems1972dissipative1, willems1972dissipative2, willems1998quadratic}, or polynomial systems \cite{zakeri2016local}. 
If one considers, for example, a chemical reactor, then the storage function and supply rate involve Gibbs free energy or entropy of multi-component chemical mixtures \cite{ramirez2013irreversible, garcia2016stability}, and such thermodynamic relations tend to be not only rough but also complicated. 
These limitations motivate data-driven estimation of dissipativity. 

\par Dissipativity estimation can be posed as a learning problem of finding the ``dissipativity parameters''\footnote{These refer to the parameters that define the storage and supply rate functions, such as the matrices in the quadratic forms or the coefficient $\beta$ in a finite-gain supply $s(y, u) = \beta\|u\|^2-\|y\|^2$).} that satisfy a linear constraint with all possible snapshots/trajectories of the system \cite{wahlberg2010non, romer2017determining, tang2019input}. 
While for linear systems, the learning is sample-efficient and one single input--output trajectory can be used \cite{berberich2020trajectory, koch2021provably, verhoek2024data}, the learning with nonlinear systems is perplexed by the need for complicated statistical methods \cite{tang2019dissipativity}, lack of behavioral theory \cite{tang2021dissipativity}, and suboptimality in trajectory sampling \cite{welikala2022line, tang2023dissipativity}, and as a result, struggles to yield an approximately correct dissipativity property \cite{locicero2024issues}. 
In Martin and Allg{\"{o}}wer \cite{martin2023data}, a relatively rigorous polynomial approximation approach was adopted, where the Taylor remainder is assumed to reside in a sector and accounted for in the estimation; on the other hand, the conservativeness of the sector bound or the adequacy of the choice of polynomial basis can be subtle. 

\par In the present paper, motivated by the wide use of Koopman operators to ``lift'' nonlinear dynamics into an infinite-dimensional spaces \cite{brunton2022modern}, we study the problems of \emph{representing the dissipative inequality} in an operator form and \emph{estimating the storage function from data} as a dissipativity certificate. 
The contributions are as follows.
\begin{enumerate}
    \item By expressing the operator model of the nonlinear system on a suitable reproducing kernel Hilbert space (RKHS), the storage and supply rate functions are considered as ``\emph{kernel quadratic forms}'' described by two self-adjoint Hilbert--Schmidt operators on the RKHS. 
    These two operators are required to satisfy a linear operator inequality with the model operators. 
    \item In order that the kernel quadratic forms are ``locally at least quadratic'' functions near the equilibrium point and hence generalize the quadratic functions, we adopt a \emph{linear--radial kernel}, which defines an RKHS spanned by \emph{component functions ($e_k: x\mapsto x_k$) multiplied by Sobolev-type functions}. 
    This construction, used in our recent works on Koopman spectrum--stability relations \cite{tang-ye2025koopman} and Lyapunov stability analysis \cite{tang-ye2026koopman} of autonomous systems, and is now used for systems with inputs. 
    \item When the supply rate is approximated as a finite-dimensional operator based on data, the estimation of storage function is then formulated as a \emph{linear matrix inequality} on data points. A statistical error bound is then derived on the violation of dissipativity, scaled by the distance or squared distance from the equilibrium point, implying the probabilistic approximate correctness of dissipativity estimation, in a pointwise sense. 
\end{enumerate}
The approach is generically applicable to nonlinear systems, as long as sufficient regularity is met by the dynamics, computationally tractable, and learns dissipativity \emph{globally} on the state--input region. 

\par The remaining paper is organized as follows. In \S\ref{sec:preliminaries}, preliminaries on dissipativity and Koopman operator modeling are provided. The formulation of dissipativity as a linear operator inequality is derived in \S\ref{sec:dissipativity}. The estimation of dissipativity from empirical data, along with statistical error bounds, are analyzed in \S\ref{sec:estimation}. Numerical studies are shown in \S\ref{sec:numerical}. Conclusions and discussions on its implications on data-driven Koopman operator-based control are found in \S\ref{sec:conclusions}.\footnote{
\emph{Notations:} We use lower-case letters to denote scalars, vectors, and functions valued in scalars or vectors, and upper-case letters for matrices and operators. 
Blackboard letters represent $\mN$ and $\mR$ in the usual sense or a region $\mbb{X}\subset \mR^d$. Calligraphic fonts such as $\mc{H}$ stand for function spaces. But conventional function/operator spaces also use regular upper-case letters, e.g., Sobolev--Hilbert function space $H^s(\mbb{X})$, space of bounded operators $B(\mc{G}, \mc{H})$, and space of Hilbert--Schmidt operators $B_2(\mc{H})$. 
Norm is denoted as $|\cdot|$ or $\|\cdot\|$, inner product -- $\langle \cdot, \cdot \rangle$, and outer product -- $\times$. By $g_1\times g_2$ we mean a rank-$1$ operator: $(g_1\times g_2)g_3 = \ip{g_2}{g_3}g_1$.  
Finally, $\mathbb{I}_n = \{0, 1, \dots, n-1\}$, $\mR_+=[0, \infty)$, and $a\vee(\wedge) b = \max(\min)\{a,b\}$. 
}

\section{PRELIMINARIES}\label{sec:preliminaries}
Consider a discrete-time nonlinear dynamical system:
\begin{equation}\label{eq:sys}
    x_{t+1} = f(x_t, u_t), \quad y_t = h(x_t, u_t),
\end{equation}
where $x_t \in \mbb{X} \subset \mR^{d_x}$ is the vector of states, $u_t \in \mbb{U} \subset \mR^{d_u}$ -- inputs, and $y_t \in \mR^{d_y}$ -- outputs. 
Without loss of generality, we let the origin be the equilibrium point of interest, i.e., $0\in \mbb{X}$, $0\in \mbb{U}$, $f(0,0) = 0$ and $h(0,0) = 0$. 

\subsection{Dissipativity}
\begin{definition}
The system \eqref{eq:sys} is said to be dissipative with respect to \emph{supply rate}\footnote{Here we consider $s$ to be dependent on $y$ and $u$. Literature (e.g., \cite{zanon2016periodic}) also defines dissipativity based on state- and input-dependent supply rate, in which case we may consider $h(x, u)=x$.} 
$s: \mR^{d_y}\times \mbb{U} \rightarrow\mR$ if there exists $v: \mbb{X} \to \mR_+$ with $v(0) = 0$, called \emph{storage function}, satisfying: 
\begin{equation} \label{eq:dissipativity}
    v(f(x,u)) - v(x) \le s(h(x,u), u), \enspace \forall (x,u) \in \mbb{X} \times \mbb{U}, 
\end{equation}
called the \emph{dissipative inequality}.\footnote{In particular, if $s(y,u) = y^\top Qy + 2y^\top Su + u^\top Ru$, then the system is $(Q,S,R)$-dissipative; if $d_y=d_u$ and $s(y,u) = y^\top u$, then the system is passive; if further $s(y,u) = y^\top u - \epsilon |y|^2$ (or $s(y,u) = y^\top u - \epsilon |u|^2$), then the system is output (input, respectively) strictly passive.} 
\end{definition}

The estimation of dissipativity of nonlinear systems provides useful information about their behaviors. For example, if the system is output strictly passive, i.e., $s(y,u) = y^\top u - \epsilon|y|^2 \leq \frac{1}{2\epsilon}|u|^2 - \frac{\epsilon}{2}|y|^2$, then %by summing the dissipative inequality from $t=0$ to $t=k-1$, one obtains % $V(x_k)-V(x_0)\leq \frac{1}{2\epsilon} \sum_{t=0}^{k-1}|u_t|^2 - \frac{\epsilon}{2} \sum_{t=0}^{k-1}|y_t|^2$, and hence 
% $\sum_{t=0}^{k-1}|y_k|^2\leq \frac{1}{\epsilon^2} \sum_{t=0}^{k-1}|u_t|^2 + \frac{2}{\epsilon}(v(x_0)-v(x_k))$, implying 
it has a finite $\ell^2$-gain of $1/\epsilon$. 
If the system is $(Q,S,R)$-dissipative and one can find an output-feedback gain $K\in \mR^{d_u\times d_y}$ such that $Q+SK+K^\top S^\top+K^\top RK \prec 0$, then $s(y,u)\leq 0$ leads to stabilization (if the system is detectable). 

For linear models $f(x,u)=Ax+Bu$ and $h(x,u)=Cx+Du$, the inequality \eqref{eq:dissipativity}, upon letting the storage function be a quadratic form $V(x)=x^\top Px$ (for a $P\succeq 0$), gives a linear matrix inequality (LMI) condition for $(Q,S,R)$-dissipativity: 
$$ \textstyle \begin{bmatrix}
    A^\top PA-P & PB \\ B^\top P & 0
\end{bmatrix} - \begin{bmatrix}
    C & D \\ 0 & I
\end{bmatrix}^\top \begin{bmatrix}
    Q & S \\ S^\top & R
\end{bmatrix}\begin{bmatrix}
    C & D \\ 0 & I
\end{bmatrix} \preceq 0. $$
While such a condition can be extended to nonlinear systems \cite{hill1976stability}, the verification would involve search for functions and become computationally intractable. 
A general perspective, as extensively used in the Koopman operator literature, is to rewrite the nonlinear dynamical and static relations as linear ones using \emph{operators}.

\subsection{Koopman Operator on Sobolev-Type RKHS}
When the inputs are absent, for an autonomous system $x_{k+1}=f(x_k)$, it is well-known that the Koopman operator is defined as the (linear) composition operator $K_f:g\mapsto g\circ f$ on any space of state-dependent functions that is invariant under composition. In the presence of inputs, we define the same composition mapping, although the domain and codomain are now different spaces.
\begin{definition}
    The \emph{Koopman operator} of system \eqref{eq:sys} is the mapping $K: \mc{G}\rightarrow K\mc{G}$, $g\mapsto g\circ f$, where $\mc{G}$ is any linear space of functions defined on $\mbb{X}$ and hence $K\mc{G}$ is a space of functions defined on $\mbb{X}\times \mbb{U}$. 
\end{definition}

In the same spirit as \cite{kohne2025error}, when the dynamics is known to be smooth to a certain order (even if $f$ is not explicitly known), then the Koopman operator can be defined on a Sobolev--Hilbert space with the matched smoothness index. This ensures a choice of $\mc{G}$ that comprises only of ``accordingly regular'' functions, despite still infinite-dimensional.
% \footnote{
%     In special cases where $f$ ``concentrates nonlinearity in a subsystem'' (for precise definitions, see Shang et al. \cite{shang2026existence}), $\mc{G}$ and $K\mc{G}$ can be chosen as finite-dimensional spaces spanned by some dictionary functions. 
%     For such systems, by defining an inner product on the dictionary, we have a finite-dimensional RKHS formulation. However, the assumption of such a finite-rank Koopman operator is restrictive, and it is doubted whether such prior knowledge is available when the model of $f$ is unknown. 
% }
\begin{lemma}
\label{lem:Koopman.Sobolev}
    Suppose that (i) $\mbb{X}\subset \mR^{d_x}$ has a Lipschitz boundary, and $\mbb{U}\subset \mR^{d_u}$ is bounded, (ii) $f\in C_\mr{b}^s(\mbb{X}\times \mbb{U}, \mbb{X})$ for some $s\in \mbb{N}$, and (iii) $\inf_{x\in \mbb{X}, u\in \mbb{U}} \lvert \mr{D}_xf(x, u) \rvert > 0$. Then $K\in B(H^s(\mbb{X}), H^s(\mbb{X}\times \mbb{U}))$. 
\end{lemma}
\begin{proof}
    For an arbitrary $g\in H^s(\mbb{X})$, we shall verify that $g\circ f\in H^s(\mbb{X}\times \mbb{U})$. Indeed, by the definition of $H^s$ spaces, 
    $$ \textstyle
        \|g\circ f\|_{H^s}^2 % &= \sum_{|\alpha|\leq s} \int_{\mbb{U}}\Bra{\int_{\mbb{X}} |\partial^\alpha (g \circ f)(x,u)|^2 \xD{x}} \xD{u} \\
        \leq c\sup_{u\in\mbb{U}} \sup_{|\alpha|\leq s} \int_{\mbb{X}} |\partial^\alpha (g \circ f)(x,u)|^2 \xD{x},$$
    in which we use $c$ to denote some positive positive constant. Due to the chain rule, for each multi-index $\alpha$, $\partial^\alpha (g \circ f)$ can be expressed as the sum of a finite number of terms, each being a product of one $\partial^\beta g$ (for some multi-index $|\beta|\leq |\alpha|$) and some $\partial^\gamma f$ (for some multi-index $|\gamma|\leq |\alpha|$). With the assumption that $f\in C_\mr{b}^s$, the squared integral of each of these terms are bounded by $\sup_{|\alpha|\leq s} \int_{\mbb{X}} |\partial^\alpha g(f(x))|^2 \xD{x}$. Then due to condition (iii), the change of variables gives a bound by $\int_{\mbb{X}} |\partial^\alpha g(x)|^2 \xD{x}$, which is linearly bounded by $\|g\|_{H^s}^2$. 
\end{proof}

\par By the Sobolev embedding theorem, if $s > d_x/2$, then $H^s(\mbb{X})$ is embedded continuously into the space of continuous functions $C^0(\mbb{X})$, and moreover, is a reproducing kernel Hilbert space (RKHS) \cite{wendland2004scattered}. 
\begin{definition}
    A RKHS with kernel $\kappa$, which is a symmetric bivariate function $\kappa\in C(\mbb{X}\times \mbb{X},\mR)$, refers to the space 
    $$\mc{H}_\kappa(\mbb{X}) = \overline{\mr{span}}\{\kappa(x,\cdot): x\in \mbb{X}\}.$$ 
    The space is a Hilbert space if for any $\{x^{(i)}\}_{i=1}^n \subset \mbb{X}$, the \emph{kernel matrix} $G = \Bra{ \kappa\bra{x^{(i)}, x^{(j)}} }_{i,j=1}^n \succeq 0$, and thus accommodate an inner product $\ip{\kappa(x,\cdot)}{\kappa(x',\cdot)} = \kappa(x,x')$ and an induced norm: $\|\kappa(x, \cdot)\| = \kappa(x,x)^{1/2}$. 
\end{definition} 
\begin{lemma}[Wendland \cite{wendland2004scattered}]
\label{lem:Koopman.RKHS}
    Let $\kappa$ be a specified by $\kappa(x, x') = \rho(|x-x'|)$ for some $\rho: \mR_+\rightarrow \mR_+$, whose Fourier transform $\hat{\rho}(\xi)$ is such that $c_1(1+|\xi|^2)^{-s}\leq |\hat{\rho}(\xi)| \leq c_2(1+|\xi|^2)^{-s}$ for some constants $c_2\geq c_1>0$ for all $\xi\in \mR$. If $\mbb{X} \subset \mR^{d_x}$ has a Lipschitz boundary, then $H^{s}(\mbb{X})$ coincides with $\mc{H}_\kappa(\mbb{X})$, with equivalent norms. 
\end{lemma}

Clearly, for $\mbb{U}\subset \mR^{d_u}$ that is also a region with Lipschitz boundary, the same can be said about $H^s(\mbb{X}\times\mbb{U}) = \mc{H}_\varkappa(\mbb{X}\times \mbb{U})$, where $\varkappa$ is the radial kernel: $\vk((x,u), (x',u'))=\rho(|(x,u)-(x',u')|)$, if further $s>(d_x+d_u)/2$. 
\begin{corollary}\label{cor:Koopman.Sobolev}
    Suppose that (i) $\mbb{X}\subset \mR^{d_x}$ and $\mbb{U}\subset \mR^{d_u}$ are bounded regions with Lipschitz boundaries, (ii) $f\in C^s(\mbb{X}\times \mbb{U}, \mbb{X})$ for some $s\in \mbb{N}$ with $s>(d_x+d_u)/2$, and (iii) $\inf_{x\in \mbb{X}, u\in \mbb{U}} \lvert \mr{D}_xf(x, u) \rvert > 0$. 
    Let $\kappa(x,x')=\rho(|x-x'|)$ and $\vk((x,u),(x',u')) = \rho(|(x,u)-(x',u')|)$ with $\rho:\mR_+\rightarrow\mR_+$, whose Fourier transform $\hat{\rho}(\xi)$ is such that $c_1(1+|\xi|^2)^{-s}\leq |\hat{\rho}(\xi)| \leq c_2(1+|\xi|^2)^{-s}$ for some constants $c_2\geq c_1>0$.
    Then $K\in B(\mc{H}_\kappa(\mbb{X}), \mc{H}_\vk(\mbb{X}\times \mbb{U}))$.\footnote{In implementation, when defining the kernel, we can use a weighted Euclidean distance with different weights on all components of $x-x'$ and $(x,u)-(x',u')$. Clearly the RKHS such defined remains equivalent.} 
\end{corollary}

On the RKHS, the reproducing property hold: $\ip{g}{\kappa(x, \cdot)} = g(x)$, $\forall g\in \mc{H}_\kappa(\mbb{X}), \forall x\in \mbb{X}$. We thus view $\kappa(x,\cdot)=:\phi_x \in \mc{H}_\kappa(\mbb{X})$ as the lifted representation, or \emph{canonical feature}, of point $x$ in the infinite-dimensional RKHS. Similarly we denote $\vk((x,u),(\cdot,\cdot))=:\vf_{(x,u)}$. 
With the canonical features, the Koopman operator $K$ can be well characterized by its \emph{adjoint operator}, namely the operator $K^\ast: \mc{H}_\vk(\mbb{X}\times \mbb{U}) \rightarrow \mc{H}_\kappa(\mbb{X})$ such that for all $g_1\in \mc{H}_\vk(\mbb{X}\times \mbb{U})$ and $g_2\in \mc{H}_\kappa(\mbb{X})$, we have $\ip{K^*g_1}{g_2}=\ip{g_1}{Kg_2}$. 
\begin{corollary}\label{cor:push.forward}
    Under the conditions in Corollary \ref{cor:Koopman.Sobolev}, $K^\ast \in B(\mc{H}_\vk(\mbb{X}\times \mbb{U}), \mc{H}_\kappa(\mbb{X}))$ satisfies
    $$K^* \vf_{(x,u)} = \phi_{x^+}, \enspace \forall (x,u)\in \mbb{X}\times \mbb{U}, $$
    in which $x^+:=f(x,u)$. 
\end{corollary}
% \par Such a relation will be used for ``lifting'' the dissipative inequality \eqref{eq:dissipativity} onto the RKHS, to be discussed in \S\ref{sec:dissipativity}, as well as learning the Koopman operator from data in \S\ref{sec:estimation}. 

\subsection{Linear--Radial Kernel and its RKHS}
\par The Sobolev-type RKHS corresponds to a translation-invariant, i.e., radial, kernel function $\kappa$. The definition of Koopman operator on such an RKHS is irrelevant to the equilibrium point at the origin. As a result, the functions contained in this RKHS must be too broad. 
In particular, the storage function $V(\cdot)$ of interest should ``look like'' a quadratic or a higher-order polynomial near the origin, which must then appear as a quadratic form of linear or higher-order functions near the origin. 
Such an intuition motivates us to restrict the function space $\mc{G}$ to contain only \emph{``locally at least linear''} functions. 

\par Let us define 
$$\textstyle \rg{H}^s(\mbb{X}) = \left\{ \sum_{k=1}^{d_x} e_kg_k: g_k \in H^s(\mbb{X})\right\}, $$
$$\textstyle \rg{C}^s(\mbb{X}\times\mbb{U}, \mbb{X}) = \left\{ \sum_{k=1}^{d_x+d_u} e_kf_k: f_k \in H^s(\mbb{X}\times\mbb{U}, \mbb{X})\right\}, $$
where we denote function $e_k(x)=x_k$ (i.e., projection onto the $k$-th component) and denote $e_kh_k$ as the pointwise function product (namely $e_kh_k:x\mapsto e_k(x)h_k(x)$). 
Thus, $\rg{H}^s(\mbb{X})$ contains state-dependent functions that ``are Sobolev and vanish at $0$'', and $\rg{C}^s(\mbb{X}\times\mbb{U}, \mbb{X})$ is the class of dynamics that ``are smooth enough and keep $0$ an equilibrium point''. % Indeed, when $s>(d_x+d_u)/2$, the Sobolev--Hilbert functions are continuous, and hence any function in $\rg{H}^s(\mbb{X})$ or $\rg{H}^s(\mbb{X}\times \mbb{U})$ vanishes at the origin.  
In our earlier works \cite{tang-ye2025koopman, tang-ye2026koopman}, the corresponding notions for autonomous systems were proposed and used for stability analysis. This work hereby extends them to input-actuated systems with no essential change in the original proof. 

\begin{lemma}[\cite{tang-ye2025koopman}]
\label{lem:linear-radial}
    Suppose that $\mbb{X}\subset \mR^{d_x}$ has Lipschitz boundary and $s>d_x/2$. Then $\rg{H}^s(\mbb{X})$ is an RKHS with the linear--radial kernel 
    $\rg{\kappa}(x,x')= (x^\top x')\kappa(x,x')$, 
    where $\kappa$ is the radial kernel with $H^s(\mbb{X})=\mc{H}_\kappa(\mbb{X})$. 
    Hence, if in addition $\mbb{U}\subset \mR^{d_u}$ has Lipschitz boundary and $s>(d_x+d_u)/2$, then $\rg{H}^s(\mbb{X}\times \mbb{U})$ is an RKHS with linear--radial kernel 
    $\rg{\vk}((x,u),(x',u'))=(x^\top x'+u^\top u') \vk((x,u),(x',u'))$, 
    where $\vk$ is the kernel for $H^s(\mbb{X}\times \mbb{U})$ as an RKHS. 
\end{lemma}
\begin{corollary}\label{cor:Koopman.LinearSobolev}
    Suppose that (i) $\mbb{X}\subset \mR^{d_x}$ and $\mbb{U}\subset \mR^{d_u}$ are bounded regions with Lipschitz boundaries, (ii) $f\in \rg C^s(\mbb{X}\times \mbb{U}, \mbb{X})$ for some $s\in \mbb{N}$ with $s>(d_x+d_u)/2$, and (iii) $\inf_{x\in \mbb{X}, u\in \mbb{U}} \lvert \mr{D}_x f(x, u) \rvert > 0$. 
    Let 
    \begin{equation}\label{eq:linear--radial}
    \begin{aligned}
        \rg\kappa(x,x')&= (x^\top x')\rho(|x-x'|) \text{ and } \\ 
        \rg\vk((x,u),(x',u')) &= (x^\top x'+u^\top u') \rho(|(x,u)-(x',u')|)
    \end{aligned}
    \end{equation}
    with $\rho:\mR_+\rightarrow\mR_+$, whose Fourier transform $\hat{\rho}(\xi)$ is such that $c_1(1+|\xi|^2)^{-s}\leq |\hat{\rho}(\xi)| \leq c_2(1+|\xi|^2)^{-s}$ for some constants $c_2\geq c_1>0$. 
    Then $K\in B(\mc{H}_{\rg\kappa}(\mbb{X}), \mc{H}_{\rg\vk}(\mbb{X}\times \mbb{U}))$.
\end{corollary}

\par In this RKHS, we denote the canonical features by $\rg\kappa(x,\cdot)=\rg\phi_x$ and $\rg\vk((x,u), (\cdot,\cdot))=\rg\vf_{(x,u)}$. While for the radial kernel, we always have $\|\phi_x\|=\sqrt{\kappa(x,x)}=\sqrt{\rho(0)}$ (a constant independent of $x$) and similarly $\|\varphi_{(x,u)}\|=\sqrt{\rho(0)}$, for the linear--radial kernel, we have, instead, 
\begin{equation}\label{eq:canonical.feature.norms}
    \begin{aligned}
        \|\rg\phi_x\| &= \sqrt{\rg\kappa(x,x)}=|x| \sqrt{\rho(0)}, \text{ and } \\
        \|\rg\vf_{(x,u)}\| &= \sqrt{\rg\vk((x,u),(x,u))} = |(x,u)| \sqrt{\rho(0)}.
    \end{aligned}
\end{equation}

\section{DISSIPATIVITY AS A LINEAR OPERATOR INEQUALITY} \label{sec:dissipativity}
Now, equipped with tools from operator theory, in this section, we show that the dissipative inequality \eqref{eq:dissipativity} becomes a linear operator inequality. As justified in the last subsection of \S\ref{sec:preliminaries}, we focus on the RKHS specified by the linear--radial kernels $\rg\kappa$ (on $\mbb{X}$) and $\rg\vk$ (on $\mbb{X}\times \mbb{U}$). 
To ensure that $K$ is well-defined and bounded, we always assume the following. 
\begin{assumption}\label{assum:regular}
    The conditions in Corollary \ref{cor:Koopman.LinearSobolev} hold.\footnote{While these condition appear to restrict the class of dynamics amenable to our treatment, it can be easily verified that all \emph{continuous-time dynamics} with sufficiently high-order smoothness, discretized with a small sampling time, satisfy the conditions (ii) and (iii). }
\end{assumption}

\subsection{Kernel Quadratic Forms} 
Intuitively, from \eqref{eq:canonical.feature.norms}, we have $\|\rg\phi_x\|^2 \propto |x|^2$ and similarly $\|\rg\vf_{(x,u)}\|^2\propto |(x,u)|^2$, and hence we consider the quadratic forms on the canonical features as generalizations of the quadratic forms in the original state and input components. 
That is, we will consider storage and supply rate functions as following ``kernel quadratic forms''. 
\begin{definition}
    A \emph{kernel quadratic form} on $\mc{H}_{\rg\kappa}(\mbb{X})$, with $\rg\kappa(x, \cdot)$ for any $x\in \mbb{X}$ denoted as $\rg\phi_x$, refers to a function $x\mapsto \ip{\rg\phi_x}{P\rg\phi_x}$ with a linear, bounded, self-adjoint operator $P$. 
    Similarly, a kernel quadratic form on $\mc{H}_{\rg\vk}(\mbb{X}\times \mbb{U})$, with canonical feature of any $(x,u)\in \mbb{X}\times\mbb{U}$ denoted as $\rg\vf_{(x,u)}$, refer to a function $(x,u)\mapsto \ip{\rg\vf_{(x,u)}}{S\rg\vf_{(x,u)}}$ with a linear, bounded, self-adjoint operator $S$. 
\end{definition}

The following two remarks give some typical examples.
\begin{remark}
    Let $v(x)=|x|^2=\sum_{k=1}^{d_x} e_k(x)^2$. Since $1\in H^s(\mbb{X})$ under our Assumption \ref{assum:regular}, $e_k\in \rg{H}^s(\mbb{X})=\mc{H}_{\rg\kappa}(\mbb{X})$. We thus have 
    $$\textstyle v(x) = \sum_{k=1}^{d_x} |\ip{e_k}{\rg\phi_x}|^2 = \ip{\rg\phi_x}{P \rg\phi_x},\, P:=\sum_{k=1}^{d_x} e_k\times e_k. $$
    Similarly, if we let $v(x)=x^\top Mx$ where $M \succeq 0$, then by finding the eigenvectors $\{w_k\}_{k=1}^{d_x}$ and eigenvalues $\{\lambda_k\}_{k=1}^{d_x}$ of $M$, we have 
    $$\textstyle v(x) = \sum_{k=1}^{d_x} |\ip{\omega_k}{\rg\phi_x}|^2 = \ip{\rg\phi_x}{P \rg\phi_x}\, P:=\sum_{k=1}^{d_x} \omega_k\times \omega_k, $$
    where $\omega_k \in \rg H^s(\mbb{X}): x\mapsto w_k^\top x$. 
\end{remark}
\begin{remark}\label{remark.2}
    Suppose that $s(y,u) = h(x,u)^\top u - \epsilon|y|^2$ ($d_y=d_u$). 
    If $h_k\in \rg H^s(\mbb{X}\times \mbb{U}) = \mc{H}_{\rg\vk}(\mbb{X}\times \mbb{U})$ for all $k=1,\dots,d_y$, then, via $u_k = \ip{e_{d_x+k}}{\rg\vf_{x,u}}$, we obtain  
    $$\textstyle s(y,u) = \sum_{k=1}^{d_y} \ip{h_k}{\rg\vf_{x,u}} \ip{e_{d_x+k}}{\rg\vf_{x,u}} - \epsilon |\ip{h_k}{\rg\vf_{x,u}}|^2.$$
    Then by letting
    $$\textstyle S = \sum_{k=1}^{d_y}\bra{ \frac{1}{2}h_k\times e_{d_x+k} + \frac{1}{2} e_{d_x+k} \times h_k - \epsilon h_k\times h_k}, $$
    we reach at $s(y,u) = \ip{\rg\vf_{x,u}}{S\rg\vf_{x,u}}$.
\end{remark}

As seen from the examples, when the supply or storage function can be written as a linear combination of squares or products of $\rg H^s$-type functions, then it is a kernel quadratic form specified by a \emph{finite-rank} adjoint operator. 
The class of finite-rank adjoint operators, however, cannot be complete. For completion, we endow the adjoint operators with a norm. 
% We say that an adjoint operator $P$ on Hilbert space $\mc{G}$ belongs to the trace class $B_1(\mc{G})$ if under any orthonormal basis $\{w_j\}_{j=1}^\infty$, it can be written as $P=\sum_{j=1}^\infty \lambda_jw_j\times w_j$ with $(\lambda_j)_{j=1}^\infty \in \ell^1$. In this case, we define $\mr{tr}\,P = \sum_{j=1}^\infty \lambda_j$. 
We say that $P$ on Hilbert space $\mc{G}$ belongs to the Hilbert--Schmidt class $B_2(\mc{G})$ if under any orthonormal basis $\{p_j\}_{j=1}^\infty$, it can be written as $P=\sum_{j=1}^\infty \lambda_j p_j\times p_j$ with $(\lambda_j)_{j=1}^\infty \in \ell^2$. In this case, we define the Hilbert--Schmidt norm $\|P\|_{B_2} = \left( \sum_{j=1}^\infty |\lambda_j|^2\right)^{1/2}$. 
The space $B_2(\mc{G})$ is then a complete, actually Hilbert, space, in which the subspace of finite-rank operators is dense \cite{lax2002functional}. % (with inner product $\ip{P}{Q}_{B_2}=\mr{tr}(PQ)$) 

% \begin{lemma}
%     Suppose $\{P_n\}_{n=1}^\infty$ be a sequence of self-adjoint Hilbert--Schmidt operators on Hilbert space $\rg H^s(\mbb{X})$ that is convergent to $P$ in $B_2(\rg H^s(\mbb{X}))$. Let $v_n(x)=\ip{\rg\phi_x}{P_n\rg\phi_x}$ and $v(x)=\ip{\rg\phi_x}{P\rg\phi_x}$. Then $\{v_n\}$ converges to $v$ pointwise, and moreover, $\{|v_n(x)-v(x)|/|x|^2\}$ converges uniformly on $\mbb{X}$. 
% \end{lemma}
% \begin{proof}
%     It is easy to verify that
%     $$\begin{aligned}
%         &\frac{|v_n(x)-v(x)|}{|x|^2} = \frac{\ip{\rg\phi_x}{(P_n-P)\rg\phi_x}}{|x|^2}
%         \\
%         &= \frac{|\mr{tr}(P_n-P)(\rg\phi_x\times \rg\phi_x)|}{|x|^2} 
%         \leq \|P_n-P\|_{B_2}\frac{\|\rg\phi_x\|^2}{|x|^2}.
%     \end{aligned}$$
%     Here we have used the identity $\mr{tr}(PQ)=\mr{tr}(QP)$ and $\|g\times g\|_{B_2}=\|g\|^2$. Since $\|\rg\phi_x\|^2=\rho(0)|x|^2$, the proof is done. 
% \end{proof}
% Clearly, similar conclusion holds for Hilbert--Schmidt kernel quadratic forms on $\rg H^s(\mbb{X}\times \mbb{U})$. 

\subsection{Dissipative Inequality}
We recall that the adjoint Koopman operator $K^*$ satisfies the relation in Corollary \ref{cor:push.forward}. Hence, if the storage function $v$ is a kernel quadratic form specified by a positive semidefinite self-adjoint operator $P$ (i.e., $\ip{g}{Pg}\geq 0$ for all $g\in \mc{H}_{\rg\kappa}(\mbb{X})$, then $v(x)=\ip{\rg\phi_x}{P\rg\phi_x}$ and hence 
$$\begin{aligned}
    v(f(x, u)) &=\ip{\rg\phi_{f(x,u)}}{P\rg\phi_{f(x,u)}} = \ip{K\rg\vf_{(x,u)}}{PK\rg\vf_{(x,u)}} \\ 
    & = \ip{\rg\vf_{(x,u)}}{K^*PK\rg\vf_{(x,u)}}. 
\end{aligned}$$
To convert $v(x)$, which is a kernel quadratic form on $\mc{H}_\kappa(\mbb{X})$, to a kernel quadratic form on $\mc{H}_{\rg\vk}(\mbb{X}\times \mbb{U})$, we need a linear operator that maps $\rg\vf_{(x,u)}$ to $\rg\phi_x$. 
To this end, define $E: \rg{H}^s(\mbb{X})\rightarrow \rg{H}^s(\mbb{X}\times \mbb{U})$, $g\mapsto Eg$, where $(Eg)(x,u)=g(x)$. In other words, $E$ is the embedding of state-dependent functions into state- and input-dependent functions. 
\begin{lemma}
    Under Assumption \ref{assum:regular}, $E \in B(\rg{H}^s(\mbb{X}), \rg{H}^s(\mbb{X}\times \mbb{U}))$, and its adjoint operator $E^*\in B(\rg{H}^s(\mbb{X}\times \mbb{U}), \rg{H}^s(\mbb{X}))$ satisfies $E^* \rg\vf_{(x,u)} = \rg\phi_x$ for all $(x, u)\in \mbb{X}\times \mbb{U}$. 
\end{lemma}

\par With the conclusion of the above lemma, we easily obtain $$v(x) = \ip{E^\ast \rg\vf_{(x,u)}}{P E^\ast \rg\vf_{(x,u)}} = \ip{\rg\vf_{(x,u)}}{EPE^\ast \rg\vf_{(x,u)}}. $$
If, in addition, the supply rate function $s(h(x,u),u)$ can also be written as a kernel quadratic form $\ip{\rg\vf_{(x,u)}}{S\rg\vf_{(x,u)}}$ with some self-adjoint operator $S$ (for which Remark \ref{remark.2} gives a typical example), then the dissipative inequality \eqref{eq:dissipativity} becomes the following linear operator inequality, imposed on all state--input pairs in $\mbb{X}\times \mbb{U}$. 

\begin{theorem}[Linear operator inequality for dissipativity]
    Suppose that Assumption \ref{assum:regular} hold. For the system \eqref{eq:sys} to be dissipative with respect to a supply rate $s(y,u) = \ip{\rg\vf_{(x,u)}}{S\rg\vf_{(x,u)}}$, it is sufficient that there exists a positive semidefinite operator $P \in B(\mc{H}_{\rg\kappa}(\mbb{X}), \mc{H}_{\rg\kappa}(\mbb{X}))$ such that $\forall (x,u)\in \mbb{X}\times \mbb{U}$, 
    \begin{equation}\label{eq:dissipativity.operator}
      \ip{\rg\vf_{(x,u)}}{\bra{KPK^*-EPE^*-S}\rg\vf_{(x,u)}}\leq 0.  
    \end{equation}
    In this case, the storage function can be chosen as $v(x)=\ip{\rg\phi_x}{P\rg\phi_x}$. 
    Moreover, if $S$ is Hilbert--Schmidt, then as long as a positive semidefinite operator solution exists for $P$, a Hilbert--Schmidt solution for $P$ exists.  
\end{theorem}
\begin{proof}
    The first half of the theorem is obvious. For the latter half of the claim, we note that if $S$ is finite-rank, say $S=\sum_{i=1}^r\omega_i\times \omega_i$, then any solution $P$ to \eqref{eq:dissipativity.matrix}, when replaced by $\Pi P\Pi^\ast$ where $\Pi$ is the projection onto $\mr{span} \{K^\ast \omega_i, E^\ast \omega_i\}_{i=1}^r$, is still a solution to \eqref{eq:dissipativity.matrix} and is finite-rank. 
    The case where $S$ is Hilbert--Schmidt can then be proved following the denseness of finite-rank operators. 
\end{proof}

\begin{remark}
    While the sufficiency of \eqref{eq:dissipativity.matrix} is obvious, the necessity requires that the storage function $v$ be indeed a kernel quadratic form. 
    Following the theory of dissipative systems \cite{lozano2013dissipative}, it is possible to construct an ``maximum available storage function'' by
    \begin{displaymath}
        \textstyle v^*(x)=\sup_{(x_0, u_0, \cdots, x_k)\in \mbb{T}(x)} \sum_{t=0}^{k-1} s(h(x_t, u_t),u_t),
    \end{displaymath} 
    where $\mbb{T}(x)$ is the collection of trajectories on which $x_0=x$ is steered to $x_k=0$ by inputs $u_0, \cdots, u_{k-1}\in \mbb{U}$ for some $k\in \mN$. 
    If for all $x$, $\mbb{T}(x)\neq \emptyset$, then $v^\ast$ is nonnegative, bounded, and satisfies the dissipative inequality. 
    For $s$ in a kernel quadratic form, the spectral decomposition $S=\sum_{i=1}^\infty \omega_i^2$ with $\omega_i\in \rg H_s(\mbb{X}\times \mbb{U})$ gives 
    $$\textstyle v^\ast(x)=\sup_{\mbb{T}(x)} \sum_{t=0}^{k-1}\sum_{i=1}^\infty \omega_i(x_t,u_t)^2. $$
    Each $\omega_i$, due to Assumption \ref{assum:regular}, as a function of $x$ must belong to $\rg H^s(\mbb{X})$. Therefore, each trajectory $\tau\in \mbb{T}(x)$ determines a sum of squares of $\rg H^s(\mbb{X})$-type functions of $x$, namely a kernel quadratic form, say $\ip{\rg\phi_x}{P_\tau \rg\phi_x}$. However, the point-wise supremum over $\tau\in \mbb{T}(x)$ may not be guaranteed to yield a kernel quadratic form. 
    At this point, the conditions that ensure $v^\ast\in \rg H^s(\mbb{X})$ remains an interesting open problem worth future studies. 
\end{remark}
\begin{remark}
    We wrote the inequality \eqref{eq:dissipativity.operator} as a \emph{pointwise inequality} at all $\rg\vf_{(x,u)}$, instead of a negative semidefiniteness constraint $KPK^*-EPE^*-S \preceq 0$. The latter is a sufficient condition, which, however, can be too strong. 
    % In the infinite-dimensional space, we still allow $KPK^*-EPE^*-S$ to have positive eigenvalues, but only restrict them to be non-consequential. 
    After all, we are concerned with dissipative inequality as an inequality over all $(x,u)$-pairs. Their canonical features $\rg\vf_{(x,u)}$ only occupy a $(d_x+d_u)$-dimensional compact manifold in the infinite-dimensional RKHS. 
\end{remark}

\subsection{Finite-Rank Approximation from Data}
\par The verification of linear operator inequality \eqref{eq:dissipativity.operator} at all $(x,u)$-pairs, however, is tractable only in a finite data-driven approximate setting. That is, we collect a sample $\mbb{S} = \{(x_i, u_i, x_i^+)\}_{i=0}^{n-1}$, where $x_i^+=f(x_i,u_i)$ for all $i\in \mbb{I}_n$, and consider \eqref{eq:dissipativity.operator} as $n$ inequality constraints on $P$ and $S$. 
In this setting, $P$ and $S$ should both be approximated as finite-rank operators, which allow matrix representations and convert \eqref{eq:dissipativity.operator} into data-based linear matrix inequalities. We will discuss such approximations and analyze the resulting statistical error bounds in the next section. 

\par % Here we show the ``matrix form'' of \eqref{eq:dissipativity.operator}. 
Let the finite-rank approximations of $S$ and a corresponding approximation of $P$ be given by: 
\begin{equation}
    % \hat{K}^\ast = \sum_{i,j=1}^n \theta^K_{ij} \rg\phi_{x_i^+} \times \rg\vf_{(x_j, u_j)}, \, 
    % \hat{E}^\ast = \sum_{i,j=1}^n \theta^E_{ij} \rg\phi_{x_i} \times \rg\vf_{(x_j^+, u_j)}, \\
    \textstyle \hat{S} = \sum_{i,j=1}^n \theta^S_{ij} \rg\vf_{(x_i, u_i)} \times \rg\vf_{(x_j, u_j)}, \, 
    \hat{P} = \sum_{i,j=1}^n \theta^P_{ij} \rg\phi_{x_i} \times \rg\phi_{x_j}. 
\end{equation}
Denote the matrices of coefficients $\Theta_S = [\theta^S_{ij}]_{i,j=1}^n$ and $\Theta_P$, and kernel matrices $G_{xx} = [\rg\kappa(x_i,x_j)]_{i,j=1}^n$, $G_{xx^+} = [\rg\kappa(x_i,x_j^+)]_{i,j=1}^n$, and $G_{xu} = [\rg\varkappa((x_i,u_i), (x_j,u_j))]_{i,j=1}^n$. 
We can verify that for each $k\in \mbb{I}_n$, 
$$\ip{\rg\vf_{(x_k,u_k)}}{K\hat{P}K^*\rg\vf_{(x_k,u_k)}} = [G_{xx^+}^\top \Theta_P G_{xx^+}]_{kk}$$ $$\ip{\rg\vf_{(x_k,u_k)}}{E\hat{P}E^*\rg\vf_{(x_k,u_k)}} = [G_{xx}^\top \Theta_P G_{xx}]_{kk},$$
$$\ip{\rg\vf_{(x_k,u_k)}}{\hat{S} \rg\vf_{(x_k,u_k)}} = [G_{xu}^\top \Theta_S G_{xu}]_{kk}.$$ 
Now the inequality of interest \eqref{eq:dissipativity.operator} is approximated as the following linear matrix inequality on $\Theta_P$ and $\Theta_S$:
\begin{equation}\label{eq:dissipativity.matrix}
    \mr{diag} \bra{G_{xx^+}^\top \Theta_P G_{xx^+} - G_{xx}^\top \Theta_P G_{xx} - G_{xu}^\top \Theta_S G_{xu}} \le 0. 
\end{equation}

\par Naturally, finite-rank approximation results in an error in dissipativity estimation. 
Suppose that $S$ is also estimated by $\hat{S}$, we solve a finite-rank operator $\hat{P}$ that satisfies \eqref{eq:dissipativity.matrix} and let $\hat{v}(x)=\ip{\rg\phi_x}{\hat{P}\rg\phi_x}$. 
When evaluating the actual dissipation:
\begin{equation}
    \begin{aligned}
        w(x, u) &= \hat{v}(f(x,u))-\hat{v}(x)-s(h(x,u),u) \\
        &= \ip{\rg\vf_{(x,u)}}{ \underbrace{(K\hat{P}K^\ast - E\hat{P}E^\ast - S)}_{=:W} \rg\vf_{(x,u)}}
    \end{aligned}
\end{equation} 
it may not be non-positive, although on the sample points, replacing $W$ by $\hat{W} = K\hat{P}K^\ast - E\hat{P}E^\ast - \hat{S}$ guarantees non-positivity. The following operator error bound is obvious. 
\begin{lemma}\label{lem:W.generalization}
    Suppose that $S$ is Hilbert--Schmidt with estimation error $\|S-\hat{S}\|_{B_2} \leq \epsilon_S$, and that $\hat{P}$ is a finite-rank operator satisfying \eqref{eq:dissipativity.matrix}. Then $\|W-\hat{W}\|_{B_2}\leq \epsilon_S$. 
\end{lemma}

\par % Based on $\epsilon_W$, when searching for the operator $\hat{P}$ under the constraints \eqref{eq:dissipativity.matrix} for a \emph{given} supply rate function, it is desirable to minimize the (squared) Hilbert--Schmidt norm of $\hat{P}$. 
The operator norm error in $W$ turns out to give a resulting bound on the square-scaled violation of dissipativity. 
\begin{corollary}\label{cor:W.scaled}
    Under the conditions of Lemma \ref{lem:W.generalization}, we have
    $$ w(x_i,u_i)/|(x_i,u_i)|^2 \leq \epsilon_S \rho(0), \enspace i\in \mbb{I}_n \text{ with } |(x_i,u_i)|\neq 0.$$
    That is, on the sample, the violation of dissipativity is bounded by at most a quadratic function of $(x,u)$. 
\end{corollary}
\begin{proof}
    Since for all $i\in \mbb{I}_n$, $\ip{\rg\vf_{(x_i,u_i)}}{\hat{W}\rg\vf_{(x_i,u_i)}} \leq 0$ holds, we have 
    $\ip{\rg\vf_{(x_i,u_i)}}{W\rg\vf_{(x_i,u_i)}} \leq \ip{\rg\vf_{(x_i,u_i)}}{(W-\hat{W})\rg\vf_{(x_i,u_i)}} \leq \|W-\hat{W}\|_{B_2} \|\rg\vf_{(x_i,u_i)}\|^2$, where, due to \eqref{eq:canonical.feature.norms}, we get $\|\rg\vf_{(x_i,u_i)}\|^2 = |(x_i,u_i)|^2 \rho(0)$. 
\end{proof}

\begin{remark}
    If the system is indeed dissipative with respect to supply rate $s(y,u) = \ip{\rg\vf_{(x,u)}}{S\rg\vf_{(x,u)}}$, then \eqref{eq:dissipativity.matrix} is feasible if $\hat{S}$ is equal to the projection of $S$ onto $\mr{span}\{\rg\vf_{(x_i,u_i)}\}_{i=1}^n$, denoted as $S_n$. However, when $\hat{S}$ is estimated with an operator error $\epsilon_S$ from $S$, i.e., a maximum error of $\epsilon_S$ from $S_n$, \eqref{eq:dissipativity.matrix} may be infeasible. Nevertheless, feasibility can be restored if we relax \eqref{eq:dissipativity.matrix} to 
    \begin{equation}\label{eq:dissipativity.matrix.relax}
    \begin{aligned}
        \Bra{G_{xx^+}^\top \Theta_P G_{xx^+} - G_{xx}^\top \Theta_P G_{xx} - G_{xu}^\top \Theta_S G_{xu}}_{ii} \leq \\ 
        \quad \epsilon_S\|\rg\vf_{(x_i,u_i)}\|^2 = \epsilon_S\rho(0)|(x_i,u_i)|^2, \enspace \forall i\in \mbb{I}_n.
    \end{aligned}
    \end{equation}
    After this, the bound guaranteed in Corollary \ref{cor:W.scaled} becomes $w(x_i,u_i)/|(x_i,u_i)|^2 \leq 2\epsilon_S \rho(0)$.
    
\end{remark}
\begin{remark}
    While Koopman operator $K$ and the embedding operator $E$ appear in the operator formulation, \emph{they do not need to be explicitly estimated from data}, as the dataset $\mbb{S}$ already covers their behavior in the recorded snapshots. 
    This not only simplifies the computation but also the eschews from Koopman operator learning, where one must calculate $G_{xx}^{-1}$ (or its regularized version) as an approximate inverse of a ``sampling operator''. The theoretical analysis of Koopman operator learning typically involves spectral theory or notions of interpolation spaces; see Kostic et al. \cite{kostic2023sharp}. 
\end{remark}

\section{DATA-DRIVEN ESTIMATION OF DISSIPATIVITY AND GENERALIZATION ERROR}\label{sec:estimation}
Now we study how $\hat{S}$ is estimated from the data sample and how the resulting $\epsilon_S$ depends on the sample size, based on which $\hat{P}$ is solved from \eqref{eq:dissipativity.matrix} or its relaxed version \eqref{eq:dissipativity.matrix.relax}. 
\begin{assumption}\label{assum:decomposition}
    The supply rate $s(h(x,u),u)$ can be written as $\sum_{j=1}^m q_j^{(0)}(x,u) q_j^{(1)}(x,u)$ where $q_j^{(0)}, q_j^{(1)}\in \rg H^s(\mbb{X}\times \mbb{U}) = \mc{H}_{\rg\vk}(\mbb{X}\times \mbb{U})$ for all $j=1,\dots,m$, where $m<\infty$.\footnote{
    As seen in Remark \ref{remark.2}, such a finite-rank assumption is appropriate and covers most practical case. If one considers a general Hilbert--Schmidt case, i.e., $s$ contains an infinite number of products of $H^s$ functions, a finite-term truncation would be needed, and the truncation error will need to be accounted for in $\epsilon_S$. 
    } 
\end{assumption}
\par Now obviously we have the true $S$ operator expressed as 
$$\textstyle S = \sum_{j=1}^m \frac{1}{2}(q_j^{(0)}\times q_j^{(1)} + q_j^{(1)}\times q_j^{(0)}).$$

\subsection{Estimation of $S$ by Kernel Ridge Regression}
\par The approximation of functions $q_j^{(0)}, q_j^{(1)}$ ($j\in \mbb{I}_m$) on the dataset $\{(x_i,u_i)\}_{i=1}^n$, with their function values at these $n$ points, boils down to a kernel-based regression problem. 
The typical formulation is a \emph{kernel ridge regression} (KRR) one, where we seek approximations of the form
$$\textstyle \hat{q}_j^{(0)} = \sum_{i=1}^n \eta_{ji}^{(0)}\rg\varphi_{(x_i, u_i)}, $$
where the coefficients $\eta_{jk}^{(0)}$ are determined by solving the following convex optimization problem:
$$\textstyle \min \, \sum_{k=1}^n \| \ip{\hat{q}_j^{(0)}}{\rg\vf_{(x_k, u_k)}} - q_j^{(0)}(x_k,u_k) \|^2 + \gamma \| \hat{q}_j^{(0)}\|_{\mc{H}_{\rg\vk}}^2. $$
The same formulation is repeated for $q_j^{(1)}$ and for all $j\in \mbb{I}_m$. The closed-form solutions for the vector of coefficients are 
\begin{equation}\label{eq:KRR}
    \eta_{j}^{(\ell)} = (G_{xu}+\beta I_n)^{-1} \Bra{q_j^{(\ell)}(x_k,u_k)}_{k=1}^n, \enspace j\in \mbb{I}_m, \ell\in\mbb{I}_2. 
\end{equation}
The resulting estimation of operator $S$ is then written as 
$$\textstyle \hat{S} = \sum_{j=1}^m \frac{1}{2}(\hat{q}_j^{(0)}\times \hat{q}_j^{(1)} + \hat{q}_j^{(1)}\times \hat{q}_j^{(0)})$$
which can then be expressed as $\hat{S} = \sum_{i,i'=1}^n \theta^S_{i,i'} \rg\vf_{(x_i,u_i)}\times \rg\vf_{(x_{i'},u_{i'})}$, where $\theta^S_{ii'} = \frac{1}{2} \sum_{j=1}^m (\eta_{ji}^{(0)} \eta_{ji'}^{(1)} + \eta_{ji}^{(1)} \eta_{ji'}^{(0)})$.  
More compactly, denote $\mr{H}^{(\ell)} = [\eta_{ji}^{(\ell)}] \in \mR^{m\times n}$, then 
\begin{equation}\label{eq:ThetaS}
    \textstyle \Theta_S = \frac{1}{2}\bra{ \mr{H}^{(0)\top}\mr{H}^{(1)} + \mr{H}^{(1)\top}\mr{H}^{(0)} }.
\end{equation}

\par As we can see, $\beta>0$ is a regularization parameter that plays the role of ensuring numerical stability and preventing overfitting. Its choice can be tuned by cross-validation. The following conclusion from Smale \& Zhou \cite[Cor. 3]{smale2007learning} gives the approximation error of KRR on RKHS, which we only specialize in our context. 
\begin{lemma}\label{lem:Smale-Zhou}
    Suppose that $\{(x_i,u_i)\}_{i=1}^n\in \mbb{X}\times \mbb{U}$ are drawn from the uniform distribution. Then there exists a constant $c>0$, such that for all $j\in \mbb{I}_m$ and $\ell\in \mbb{I}_2$, we have
    $$ \|q_j^{(\ell)} - \hat{q}_j^{(\ell)}\|_{\mc{H}_{\rg\vk}} \leq cn^{-1/3}\log(2/\delta)$$
    with probability $1-\delta$ over random sampling.  
\end{lemma}
\begin{corollary}\label{cor:W.scaled.statistical}
    Under the condition of Lemma \ref{lem:Smale-Zhou} and Assumption \ref{assum:decomposition}, there exists a constant $c>0$ such that 
    $$\|S-\hat{S}\|_{B_2} \leq cn^{-1/3}\log (2m/\delta) $$
    with probability $1-\delta$ over random sampling. 
\end{corollary}

\par Combining Corollaries \ref{cor:W.scaled} and \ref{cor:W.scaled.statistical}, with probability $1-\delta$, we have $w(x_i,u_i)/|(x_i,u_i)|^2 \leq cn^{-1/3}\log (2m/\delta)$ for all $i\in \mbb{I}_n$ with $|(x_i,u_i)|\neq 0$, where $c>0$ is a constant.

\subsection{Generalization Error Bound}
\par At the end, we should generalize the scaled violation of dissipativity to all points $(x,u)\in \mbb{X}\times \mbb{U}$. 
With $W$ being a self-adjoint Hilbert--Schmidt operator, the kernel quadratic form that it specifies, if scaled by $|(x,u)|$ near the origin, or $|(x,u)|^2$ away from the origin, should be a function with nice continuity property. Hence, when the data sample is rich, the bound in Corollary \ref{cor:W.scaled} can be extended from the sample points onto the whole region ``without much inflation''. 
\begin{lemma}
    Suppose that Assumption \ref{assum:regular} hold and that $W$ is a self-adjoint and Hilbert--Schmidt operator on $\mc{H}_{\rg\vk}(\mbb{X}\times \mbb{U})$. Let $\overline{w}_1(x,u) = w(x,u)/|(x,u)|$ and $\overline{w}_2(x,u) = w(x,u)/|(x,u)|^2$, and $r>0$ small enough such that $\mbb{B}_r :=\{(x, u)\in \mR^{d_x+d_u}: |(x,u)|< r \} \subset \mbb{X}\times\mbb{U}$. 
    Then $\overline{w}_1\in C^{\alpha}(\mbb{B}_r, \mR)$, the class of continuous functions with H{\"{o}}lder exponent $\alpha$, with $\alpha= 1\wedge (s-(d_x+d_u)/2)$, while $\overline{w}_2\in C^\alpha((\mbb{X}\times\mbb{U}) \backslash \mbb{B}_r, \mR)$. 
    The modulus of H{\"{o}}lder continuity can be chosen as $\propto r^{-1}$. 
\end{lemma}
\begin{proof}
    Due to the Sobolev embedding theorem, $H^s \hookrightarrow C^\alpha$. We claim that any function of the form 
    $\textstyle \omega(z)= \frac{1}{|z|} \left( \sum_{k=1}^d z_k\omega_k(z) \right)^2, $
    where all $\omega_1, \cdots, \omega_k\in H^s(\mbb{X}\times \mbb{U})$, is $C^\alpha$ at the origin. 
    This can be shown by the Lipschitz continuity of the $z_kz_l/|z|$ term and the H{\"{o}}lder continuity of the product term $\omega_k\omega_l$. 
    \par Now $W$ is Hilbert--Schmidt, $w(x,u) = \ip{\rg\vf_{(x,u)}}{W\rg\vf_{(x,u)}}$ can be written as $\sum_{i=1}^\infty \left( \sum_{k=1}^{d_x+d_u} z_k\omega_{ik}(z) \right)^2$, with the squared $H^s$-norms of $\omega_{ik}$ summable. Hence, $\overline{w}_1\in C^\alpha$ in $\mbb{B}_r$. 
    Since $\mbb{X}$ and $\mbb{U}$ are bounded, outside of $\mbb{B}_r$, $1/|z|^2$ is Lipschitz continuous. By elementary calculation similar to the proof for $\overline{w}_1$, it can be shown that $\overline{w}_2$ remains in $H^s$, and hence $C^\alpha$. 
\end{proof}

\begin{theorem}[Generalized error of dissipativity violation]
\label{th:generalization}
    Suppose that Assumptions \ref{assum:regular} and \ref{assum:decomposition} hold. Then $\exists c>0$, such that at large $n$, with probability $1-\delta$, we have
    \begin{equation}\textstyle
        \frac{w(x,u)}{|(x,u)|^2 \vee |(x,u)|} \leq c\bra{\bra{\frac{\log(n/\delta)}{n}}^{\alpha^*} + \frac{1}{n^{1/3}} \log\frac{4m}{\delta}} 
    \end{equation}
    for all $(x,u)\in \mbb{X}\times\mbb{U}$, where $\alpha^* = \frac{1}{d_x+d_u} \wedge (\frac{s}{d_x+d_u}-\frac{1}{2})$. 
\end{theorem}
\begin{proof}
    Denote $\eta = \sup_{(x,u)\in \mbb{X}\times \mbb{U}} \min_{i\in \mbb{I}_n} |(x,u)-(x_i,u_i)|$ as the fill distance of the sample. 
    Let $r>\eta$ be such that $\mbb{B}_{r+\eta}\subset \mbb{X}\times \mbb{U}$. For any $(x,u)$, if $|(x,u)|<r$ then there exists a sample point $(x_i,u_i)\in \mbb{B}_{r+\eta}$ with distance not exceeding $\eta$, and hence due to the H{\"{o}}lder continuity of $\overline{w}_1$, $w(x, u)/|(x,u)|\leq \epsilon_S\rho(0)+c_1\eta^\alpha$; if $|(x,u)|\geq r$ then there exists a $(x_i,u_i)$ with $|(x_i,u_i)|\geq r-\eta$, and hence due to the H{\"{o}}lder continuity of $\overline{w}_1$, $w(x, u)/|(x,u)|^2 \leq \epsilon_S\rho(0)+c_2\eta^\alpha/(r-\eta)$. Here $c_1$ and $c_2$ are constants. In both cases, the right-hand-side error bound is bounded by a linear combination of $\eta^\alpha$ and $\epsilon_S$. 
    Since $\eta \leq c_\eta(\log (n/\delta)/n)^\alpha$ with probability $1-\delta/2$ for some $c_\eta>0$ \cite{vershynin2018high}, and $\epsilon_S$ was given in Lemma \ref{cor:W.scaled.statistical}, we obtain the desired conclusion. 
\end{proof}

\section{NUMERICAL EXPERIMENTS}\label{sec:numerical}
To validate the proposed data-driven dissipativity estimation, we evaluate its performance with three case studies. 
% We compare the learned storage functions with the quadratic storage based on local linear model to show its capability to handle nonlinear systems, and remark that the standard radial kernel is unable to give a meaningful storage function such that $v(0)=0$. 
% In all examples, the hyperparameters on the kernel functions and KRR regularization parameters in \eqref{eq:KRR} are tuned by cross-validation using the mean-squared error metric. 
In all experiments, the storage function is solved under constraint \eqref{eq:dissipativity.matrix} with a minimization objective of $\mr{tr}(G_{xx}^\top P G_{xx})$ (sum of the storage values at sample points), using \texttt{cvxpy}.

\subsection{Case 1: Polynomial System with Analytical Storage}
To show the precision of the proposed method, we consider a synthetic polynomial system:
$$\textstyle x_1^+ = \frac{1}{2}\left( x_2 + x_1^2 \right), 
\enspace x_2^+ = u - \frac{1}{4} \left( x_2 + x_1^2 \right)^2, 
\enspace y = x_1$$
and aim to find a storage function that verifies the rate $s(y,u) = \frac{1}{4}u^2 - y^2$ (which implies $\ell^2$-gain). A ground-truth storage function that fits this given $s$, and actually satisfies exactly $v(x^+)-v(x)=s(y,u)$, is the quartic polynomial:
$$\textstyle v^*(x) = x_1^2 + \frac{1}{4}(x_2 + x_1^2)^2.$$ 
We sample $8$ trajectories, each with a length of $10$ time steps, initial state randomly located in $\mathcal{X} = [-2, 2] \times [-2, 2]$ and inputs randomly sampled from $u_t \in [-1, 1]$. We examine (i) whether $v_\ast$ can be learned precisely from data under the specified $s$ with the proposed approach, (ii) if $v(x)$ is restricted to be a simple quadratic form $v(x)=x^\top Mx$, whether finite $\ell^2$-gain property can still be verified from data, and (iii) if we assume an ``oracle lifting'' $\psi(x) = (x_1, x_2 + x_1^2)$ and restrict $v(x) = \psi(x)^\top M\psi(x)$, whether the estimated $v$ appear close to the ground truth. 

\par The ground-truth $v^\ast$ and the estimates from three experiments in order are shown in Fig.~\ref{fig:case1_storage}. 
It can be claimed that our linear--radial kernel-based approach precisely learned the storage function (depicted as $\hat v_1(x)$), close to the case as if we had an ``oracle'' that $v$ is a quadratic form of $\psi(x)=(x_1, x_2+x_1^2)$ (depicted as $\hat{v}_3(x)$). 
In contrast, the storage function restricted to a simple quadratic form ($\hat{v}_2(x)$) has a completely different shape. 
The foregoing results are for the case $\beta = 1/4$. 
In addition, numerical simulations verify there is no feasible storage function if we set $\beta < 1/4$, which is consistent with the model.  
 
\begin{figure}[!t]
    \centering
        \includegraphics[width=\columnwidth]{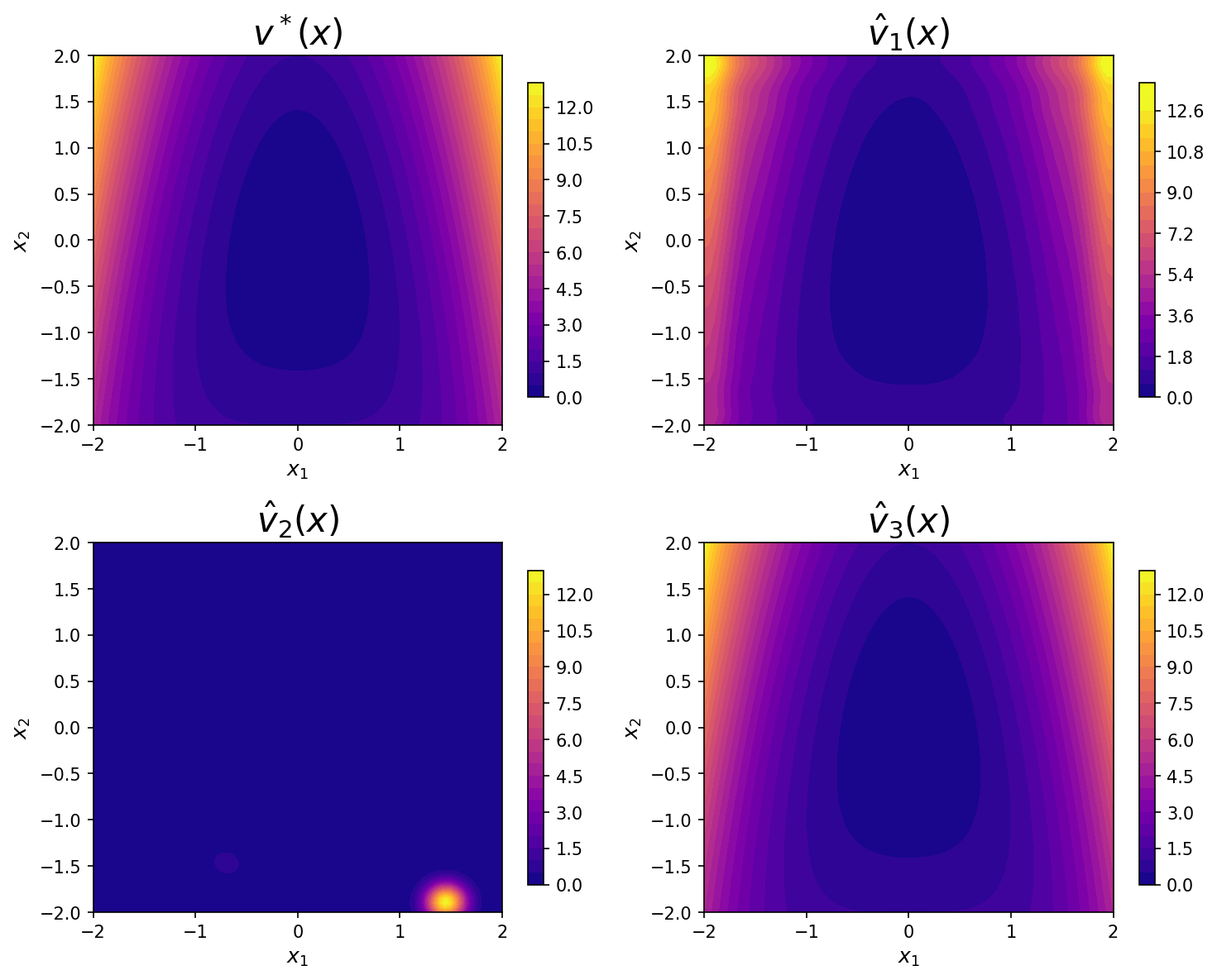}
    \caption{Storage function estimates for Case 1. }
    \label{fig:case1_storage}
    \vspace{-1.5em}
\end{figure}

\subsection{Case 2: Inverted Pendulum}
Consider an inverted pendulum with continuous-time dynamics, discretized by a Runge--Kutta 4th-order method with sampling interval $\Delta t = 0.05$: 
$$\dot{x}_1 = x_2, \enspace\dot{x}_2 = \sin x_1 - \gamma(x_2) + u, \enspace y_1 = x_1, \enspace y_2=x_2, $$
where $x_1$ is the angle, $x_2$ is the angular velocity, and $u$ is the external torque. 
The friction $\gamma(x_2)$ is assumed to be unknown, but of sectorial nonlinearity: $x_2\gamma(x_2) \geq ax_2^2$ for some $a>0$. 
We consider supply rate of the type $s(x,u) = \beta \sin^2 y_1 + y_2u$, and aim to find the lowest $\beta$ under which there still exists a feasible storage function dissipative with respect to $s$. 
When choosing this form of $s$, we view $y_2u$ as the power provided by the external torque and $\sin^2 y_1$ as a term associated with the instability of the system. 
To understand the theoretical lower bound on $\beta$, consider a candidate storage $v(x) = 1 - \cos x_1 + \frac{1}{2}x_2^2$, for which
$$\begin{aligned}
    \dot{v}(x) &= 2x_2\sin x_1 - x_2\gamma(x_2) + x_2u \\
    & \leq 2x_2\sin x_1 -ax_2^2 + x_2u  \leq (1/a)\sin^2 x_1 + x_2u. 
\end{aligned}$$
Hence, all $\beta\geq 1/a$ are valid.
When generating the simulation data, we use $\gamma(x_2)=x_2(1+x_2\tanh x_2)/6$, which verifies $a=1/6$. 

\begin{figure}[!t]
    \centering
        \centering
        \includegraphics[width=\columnwidth]{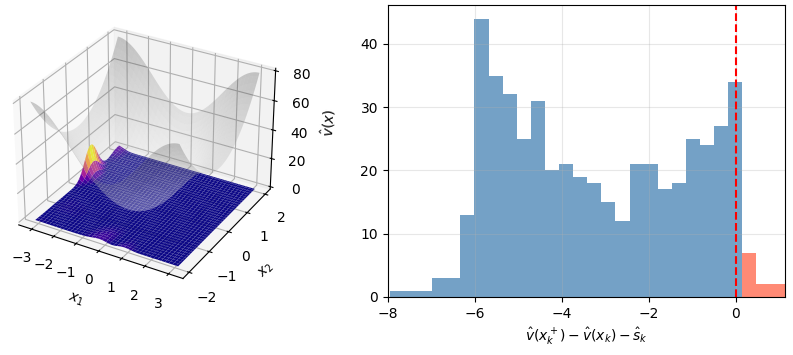}
        \caption{Storage function estimate (left) and dissipativity satisfaction check (right) for Case 2 under $\beta = 6$.}
    \label{fig:case2_storage}
    \vspace{-1.5em}
\end{figure}

\par We sample $75$ snapshots, with the state uniformly sampled from $\mathcal{X} = [-\pi, \pi] \times [-2, 2]$ and the input $u \in [-0.5, 0.5]$. Setting $\beta=6$, dissipativity learning returns a feasible storage function. 
The evaluation of the dissipative inequality on $500$ new snapshots sampled from the same distribution validates the result, with few violations ($4.2\%$). 
The learned $v$ profile, as plotted in Fig.~\ref{fig:case2_storage}, however, is very different from the ``physically intuitive'' storage $v(x) = 1-\cos x_1 + \frac{1}{2}x_2^2$ (plotted in the gray shadow), which is the mechanical energy of the ``non-inverted pendulum'' (scaled by $\Delta t$). 
This shows that the learned storage function is \textit{drastically less conservative}.

\subsection{Case 3: Bioreactor}
Finally, we consider a bioreactor \cite[Problem~II.23]{romagnoli2020introduction}:
$$\textstyle \dot{x}_1= (\mu(x_2) -1-u)x_1, 
\text{ where } \mu(x_2)=\frac{1+x_2}{1+0.7x_2+0.05x_2^2},$$
$$\textstyle \dot{x}_2 = (1+u)(20-x_2) - 20(1+x_1)\mu(x_2), \enspace
y=x_2.$$
The two states represent scaled deviations from steady-state biomass and substrate concentrations, and the $\mu$ function represents the biomass growth rate. The input $u$ is a dilution rate. Here, physical intuition is not enough to readily give us a guess on the storage or supply rate function. 
% We impose structural assumptions to obtain a tractable finite-dimensional optimization problem. Specifically, the storage function is restricted to a quadratic form $v(x) = x^\top M x$ and exclude the trivial solution $M = 0$. 
We investigate whether the system is $(Q,S,R)$-dissipative with respect to any supply rate in the form of $s(y, u) = qy^2 + yu$ for some $q<0$, and we seek to find a choice of $q$ so that the violation of dissipativity, $w(x, u)$, is empirically small.

\par In the case study, we sample $100$ snapshots over $\mbb{X} = [-0.5, 0.5] \times [-0.5, 0.5]$ and $\mbb{U} = [-0.25, 0.25]$. A feasible storage function found and the dissipativity validation results are shown in Fig.~\ref{fig:fig3}, for $q=2.0$. 
We discover an empirical storage function with nontrivial profile. The violation to the dissipative inequality $w(x,u)$, which is positive on $5.6\%$ of the sample points as plotted in Fig.~\ref{fig:fig3}, is bounded apparently by $|(x,u)|^2$ away from the origin. This observation confirms the conclusion of Theorem \ref{th:generalization}.  
\begin{figure}[t]
    \centering
        \centering
        \includegraphics[width=\columnwidth]{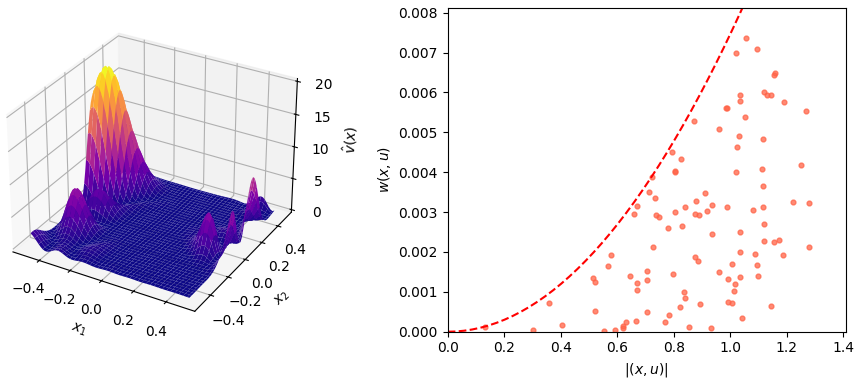}
       \caption{Storage function estimate and distribution of dissipativity violations for Case 3.}
    \label{fig:fig3}\vspace{-1.5em}
\end{figure}

%%%%%%%%%%%%%%%%%%%%%%%%%%%%%%%%%%%%%%%%%%%%%%%%%%%%%%%%%%%%%%%%%%%%%%%%%%%%%%%%
\section{CONCLUSIONS}\label{sec:conclusions}
\par A data-driven dissipativity estimation method for nonlinear systems, based on operator theory and approximated with data, is presented. 
Building on the tool of linear--radial kernels, the dissipative inequality is estimated under linear matrix inequality constraints, which ensures the dissipativity on the sample points in a bounded region and can be generalized to all points therein, with bounded violation of dissipativity (scaled by distance or squared distance to the equilibrium point) that vanishes at large data limit. 
The learning of dissipativity can be used in multiple potential ways, e.g., to characterize plant-model mismatch in robust control problems, to certify stability, and to design dissipativity-based controllers in a model-free approach. 

\par An existing limitation of this work is the dependence on state data, which may not be accessible when the states are not measured fully. 
This is in fact a common issue in Koopman operator-related works. Overcoming it would require data-driven state observer synthesis and even \emph{operator modeling directly from input--output data}. These important problems are worth examination by researchers.  

\bibliographystyle{ieeetr}
\bibliography{mybib}

\end{document}